\documentclass[conference]{IEEEtran}
\IEEEoverridecommandlockouts
\usepackage{cite}
\usepackage{amsmath,amssymb,amsfonts}
\usepackage{algorithmic}
\usepackage{graphicx}
\usepackage{textcomp}
\usepackage{xcolor}
\usepackage{microtype}
\usepackage{bm}
\usepackage{amsthm}
\usepackage{url}
\usepackage{hyperref}

\newtheorem{theorem}{Theorem}
\newtheorem{lemma}{Lemma}
\newtheorem{definition}{Definition}

\newcommand{\remove}[1]{}

\newcommand{\Ra}{\Rightarrow}
\newcommand{\ra}{\rightarrow}
\newcommand{\CC}{{L}}
\newcommand{\CB}{{L_B}}
\newcommand{\C}{G}
\newcommand{\RR}{\mathbf{R}}

\DeclareMathOperator{\frontier}{frontier}
\DeclareMathOperator{\forbidden}{forbidden}

\usepackage{microtype}

\begin{document}
\title{Applying Predicate Detection to the Constrained Optimization Problems
\thanks{
Supported by NSF CNS-1812349, CNS-1563544, and the Cullen Trust for Higher Education Endowed Professorship}
}
\author{Vijay K. Garg,\\
  The University of Texas at Austin,\\
  Department of Electrical and Computer Engineering,\\
  Austin, TX 78712, USA}

\bibliographystyle{plainurl}


\maketitle

\begin{abstract}
We present a method to design parallel algorithms for  constrained combinatorial optimization problems. Our method solves and generalizes many
classical combinatorial optimization problems including the stable marriage problem, the shortest path problem and the market clearing price problem.
These three problems are solved in the literature using Gale-Shapley algorithm, Dijkstra's algorithm, and 
Demange, Gale, Sotomayor algorithm. Our method solves all these problems by casting them as searching for an element that 
satisfies an appropriate predicate in a distributive lattice. Moreover, it solves generalizations of all these problems  - namely finding the optimal solution
satisfying additional constraints called {\em lattice-linear} predicates.
For stable marriage problems, an example of such a constraint is that Peter's regret is less than that of Paul.
For shortest path problems, an example of such a constraint is that cost of reaching vertex $v_1$ is at least the cost of reaching vertex $v_2$.
For the market clearing price problem, an example of such a constraint is that $item_1$ is priced at least as much as $item_2$.
In addition to finding the optimal solution,  our method is useful in enumerating all constrained stable matchings, and all constrained market clearing price vectors.
 \end{abstract}


\section{Introduction}
We present a method called {\em lattice-linear
predicate detection} that can solve many combinatorial optimization problems. We use this method to solve
generalization of three of the most fundamental problems in combinatorial optimization --- the stable marriage problem  \cite{gale1962college}, 
the shortest path problem \cite{Dijkstra1959}, and the assignment problem \cite{munkres1957algorithms}. 
Due to the importance and applications of these problems, each one of them
has been the subject of numerous books and thousands of papers. The classical algorithms to solve these problems are
Gale-Shapley algorithm for the stable marriage problem  \cite{gale1962college}, Dijkstra's algorithm for the shortest path problem \cite{Dijkstra1959}, and Kuhn's Hungarian method
to solve the assignment problem \cite{munkres1957algorithms} (or equivalently, Demange, Gale, Sotomayor auction-based algorithm \cite{demange1986multi} for market clearing prices).
Could there be a single efficient algorithm that solves all of these problems? 

In this paper, we describe a technique that solves not only these problems but more {\em  general} versions
of each of the above problems. We seek the optimal solution for these problems that satisfy additional constraints
modeled using a {\em lattice-linear} predicate \cite{chase1998detection}. When the set of constraints is empty,
we get back the classical problems.

Our technique requires the underlying search space to be viewed as a distributive lattice \cite{Birk3, davey}.
Common  to all these seemingly disparate combinatorial optimization problems is the
structure of the {\em feasible} solution space. 
The set of all stable matchings,  the set of all feasible rooted trees for the shortest path problem, and
the set of all market clearing prices are all closed under the meet operation of the lattice.
If the order is appropriately defined, then finding the optimal
solution (the man-optimal stable marriage, the shortest path cost vector, the minimum market clearing price vector) is equivalent
to finding the infimum of all feasible solutions in the lattice.

We note here that it is well-known that the set of stable matching and the set of market clearing price vectors form distributive 
lattices.
The set of stable matchings forms a distributive lattice is given in \cite{knuth1997stable} where the result is attributed to Conway.
The set of market clearing price vectors forms a distributive lattice is given in \cite{shapley1971assignment}.
However, the algorithms to find the man-optimal stable matching and the minimum market clearing price vectors are not derived 
from the lattice property. 
In our method, once the lattice-linearity of the feasible solution space is established, the
algorithm to find the optimal solution falls out as a consequence. To the best of our knowledge, this is the first paper to 
derive Gale-Shapley's algorithm, Dijkstra's algorithm and Demange-Gale-Sotomayor's algorithm from a single algorithm
by exploiting a lattice property.

The lattice-linear predicate detection method to solve the combinatorial optimization problem is as follows. 
The first step is to define a lattice of vectors $\CC$ such that each vector is {\em assigned} a point in the search space.
For the stable matching problem, the vector corresponds to the assignment of men to women (or equivalently, the choice number for
each man).
For the shortest path problem, the vector assigns a cost to each node.
For the market clearing price problem, the vector assigns a price to each item.
The comparison operation ($\leq$) is defined on the set of vectors such that the least vector, if feasible, is the extremal solution of interest.
For example, in the stable marriage problem if each man orders women according to his preferences and
every man is assigned the first woman in the list, then this solution is the man-optimal solution whenever the assignment is a matching and has no blocking pair.
Similarly, in the shortest path problem and the minimum market clearing price problem,  the zero vector would be optimal if it were feasible.

The second step in our method is to define a boolean predicate $B$ that models feasibility of the vector.
For the stable matching problem, an assignment is feasible iff it is a matching and there is no blocking pair.
For the shortest path problem, an assignment is feasible iff there exists a rooted spanning tree at the source vertex such that
the cost of each vertex is greater than the cost of traversing the path in the rooted tree.
For the minimum market clearing price problem, a price vector is feasible iff it is a market clearing price vector.

The third step is to show that the feasibility predicate is a lattice-linear predicate \cite{chase1998detection}.
Lattice-linearity property allows one to search for a feasible solution efficiently. If any point in the search space is not feasible,
it allows one to make progress towards the optimal feasible solution without any need for exploring
multiple paths in the lattice. Moreover, multiple processes can make progress towards the feasible solution independently.
In a finite distributive lattice, it is clear that the maximum number of such advancement steps before
one finds the optimal solution or reaches the top element of the lattice is equal to the height of the lattice.
Once this step is done, we get the following outcomes. 

First, 
 by applying the lattice-linear predicate detection algorithm to unconstrained
problems, we get
Gale-Shapley algorithm for the stable matching problem, Dijkstra's algorithm for the shortest path problem and Demange, Gale, Sotomayor's algorithm for
the minimum market clearing price. In fact, the lattice-linear predicate detection method yields a parallel version of these algorithms and by restricting these to
their sequential counterparts, we get these classical sequential algorithms.

Second, we get optimal solutions
for the constrained version of each of these problems, whenever the constraints are lattice-linear.
We solve the {\em Constrained Stable Matching Problem} where in addition to men's preferences and women's preferences, there may be a set of
lattice-linear 
constraints.
For example, we may require that Peter's regret  \cite{gusfield1989stable} should be less than that of Paul,
where the {\em regret} of a man in a matching is the choice number he is assigned. 
We note here that some special cases of the constrained stable marriage problems have been studied.
Dias et al \cite{Dias2003,Cseh2016} study the stable marriage problem with restricted pairs. A restricted pair is either a {\em forced} pair which is required to be in the
matching, or a {\em forbidden} pair which must not be in the matching. Both of these constraints are {\em lattice-linear}  and therefore can be modeled in our system.
The constrained shortest path problem asks for a rooted tree at the source node with the smallest cost at each vertex that satisfies additional constraints
of the form ``the cost of reaching node $x$  is at least the cost of reaching node $y$'', 
 ``the cost 
of reaching $x$ must be equal to the cost of reaching $y$'',  and ``the cost of reaching $x$ must be within $\delta$ of the cost of reaching $y$''.
For the market clearing price problem, we consider constraints on the clearing prices of the form that item $i$ must be priced at least as much as item $j$, or
the difference in prices for item $i$ and $j$ must not exceed $\delta$.

Third, by applying a constructive version of Birkhoff's theorem on finite
distributive lattices \cite{Birk3,davey}, we give an algorithm that outputs a succinct representation of all feasible solutions.
In particular, the join-irreducible elements \cite{davey} of the feasible sublattice can be determined efficiently (in polynomial time).
For the constrained stable matching problem, we get a concise representation of all stable matchings
that satisfy given constraints.
Thus, our method yields a more general version of 
rotation posets \cite{gusfield1989stable} to represent all {\em constrained}  stable matchings.
Analogously, we get a concise representation of all constrained integral market clearing price vectors.

The paper is organized as follows.
Section \ref{sec:lattice-linear} defines lattice-linear predicates and gives a simple parallel algorithm called $LLP$ (Lattice-Linear Predicate detection algorithm).
Section \ref{sec:csmp} uses $LLP$ to solve the constrained stable matching problem.
Section \ref{sec:csssp} solves the constrained single source shortest path problem, and
Section \ref{sec:cmcp} solves the constrained market clearing prices problem.
Section \ref{sec:slice} gives a concise representation of all feasible solutions for the constrained stable matching and market clearing price problem.
Finally, Section \ref{sec:conc} gives conclusions and future work.

\section{Lattice-Linear Predicates}
\label{sec:lattice-linear}
Let $\CC$ be the lattice of all $n$-dimensional vectors of reals greater than or equal to zero vector and less than or equal to a given vector $T$
where the order on the vectors is defined by the component-wise natural $\leq$.
The minimum element of this lattice is the zero vector.
The lattice is used to model the search space of the combinatorial optimization problem.
For simplicity, we are considering the lattice of vectors of non-negative reals; later we show that our results are applicable to
any distributive lattice.
The combinatorial optimization problem is modeled as finding the minimum element in $\CC$ that satisfies a boolean {\em predicate} $B$, where
$B$ models {\em feasible} (or acceptable solutions).
We are interested in parallel algorithms to solve the combinatorial optimization problem with $n$ processes.
We will assume that the systems maintains as its state the current candidate vector $\C \in \CC$ in the search lattice, 
where $\C[i]$ is maintained at process $i$. We call $\C$, the global state, and $\C[i]$, the state of process $i$.

Finding an element in lattice that satisfies the given predicate $B$, is called the {\em predicate detection} problem.
Finding the {\em minimum} element that satisfies $B$ (whenever it exists) is the combinatorial optimization problem.
We now define {\em lattice-linearity} which enables efficient computation of this minimum element.
Lattice-linearity is first defined in  \cite{chase1998detection} in the context of detecting global conditions in a distributed system where it is simply called linearity. We use the term
{\em lattice-linearity} to avoid confusion with the standard usage of linearity.

A key concept in deriving an efficient  predicate detection algorithm is that of a {\em
forbidden} state.  
Given a predicate $B$, and a vector $\C \in \CC$, a state $G[i]$ is {\em forbidden} (or equivalently, the index $i$ is forbidden) if 
for any vector $H \in \CC$ , where $G \leq H$, if $H[i]$ equals $\C[i]$, then $B$ is false for $H$.
Formally,
\begin{definition}[Forbidden State  \cite{chase1998detection}]
  Given any distributive lattice $\CC$ of  $n$-dimensional vectors of $\RR_{\ge 0}$, and a predicate $B$, we define
  $ \forbidden(G,i,B) \equiv \forall H \in \CC : G \leq H : (G[i] = H[i]) \Rightarrow
  \neg B(H).$
\end{definition}


We define a predicate $B$ to
be {\em lattice-linear} with respect to a lattice $\CC$
 if for any global state $G$,  $B$ is false in $G$ implies that $G$ contains a
{\em forbidden state}. Formally,
\begin{definition}[lattice-linear Predicate  \cite{chase1998detection}]
A boolean predicate $B$ is {\em {lattice-linear}} with respect to a lattice $\CC$
iff
$\forall G \in \CC: \neg B(G) \Ra (\exists i: \forbidden(G,i,B))$.
\end{definition}

We now give some examples of lattice-linear predicates.
Our first example relates to scheduling of $n$ jobs. Each job $j$ requires time $t_j$ for completion and has a set of
 prerequisite jobs, denoted by $pre(j)$, such that it can be started only after all its prerequisite jobs
have been completed. Our goal is to find the minimum completion time for each job.
We let our lattice $\CC$ be the set of all possible completion times. A completion vector $\C \in \CC$ is feasible iff $B_{jobs}(\C)$ holds where
$B_{jobs}(\C) \equiv \forall j: (\C[j] \geq t_j) \wedge (\forall i \in pre(j): \C[j] \geq \C[i] + t_j)$.
$B_{jobs}$ is lattice-linear because if it is false, then there exists $j$ such that 
either $\C[j] < t_j$ or $\exists i \in pre(j): \C[j] < \C[i]+t_j$. We claim that $\forbidden(\C, i, B_{jobs})$. Indeed, any vector $H \geq \C$ cannot
be feasible with $\C[j]$ equal to $H[j]$. The minimum of all vectors that satisfy feasibility corresponds to the minimum completion time.


Our second example relates to (exclusive) prefix sum of an array $A$ with non-negative reals. We are required to output an array $\C$ such that
$\C[j]$ equals sum of all entries in $A$ from $0$ to $j-1$. We define $\C$ to be feasible iff $B_{prefix}$ holds where
$B_{prefix} \equiv (\forall j > 0): (\C[j] \geq \C[j-1] + A[j-1])$. Again, it is easy to verify that
$B_{prefix}$ is lattice-linear. The minimum vector $\C$ that satisfies $B_{prefix}$ corresponds to the exclusive prefix sum of the array $A$.

As an example of a predicate that is not lattice-linear, consider the predicate $B \equiv \sum_j \C[j] \geq 1$ defined on the space of 
two dimensional vectors. Consider the vector $\C$ equal to $(0,0)$. The vector $\C$ does not satisfy $B$. For $B$ to be lattice-linear
either the first index or the second index should be forbidden. However, 
none of the indices are
forbidden in $(0,0)$. The index $0$ is not
forbidden because the vector $H = (0,1)$ is greater than $\C$, has $H[0]$ equal to $\C[0]$ but it still satisfies $B$.
The index $1$ is also not forbidden
because $H =(1,0)$ is greater than $\C$, has $H[1]$ equal to $\C[1]$ but it satisfies $B$.

The following Lemma is useful in proving lattice-linearity of predicates.
\begin{lemma}\label{lem:lin}
Let $B$ be any boolean predicate defined on a lattice $\CC$ of vectors. \\
(a) Let $f:\CC \ra \RR_{\ge 0}$ be any monotone function  defined on the lattice $\CC$ of vectors of $\RR_{\ge 0}$.
Consider the predicate
$B \equiv \C[i] \geq f(\C)$ for some fixed $i$. Then, $B$ is lattice-linear.\\
(b)  Let $\CB$ be the subset of the lattice $\CC$
of the elements that satisfy $B$.
Then, $B$ is lattice-linear iff $\CB$ is closed under meets.\\
(c) If $B_1$ and $B_2$ are lattice-linear then $B_1 \wedge B_2$ is also lattice-linear.
\end{lemma}
\begin{proof}
(a) Suppose $B$ is false for $\C$. This implies that $\C[i] < f(\C)$. Consider any vector $H \geq \C$ such that $H[i]$ is equal to $\C[i]$.
Since $\C[i] < f(\C)$, we get that $H[i] < f(\C)$. 
The monotonicity of $f$ implies that $H[i] < f(H)$ which shows that $\neg B(H)$.\\
(b) This is shown in \cite{chase1998detection}.  Assume that $B$ is not
  lattice-linear.  This implies that there exists a global state $G$ such that $\neg
  B(G)$, and $\forall i: \exists H_i \geq G: (G[i] = H_i[i])$ and
  $B(H_i)$.  Consider
  $Y = \cup_i \{ H_i \}.$
  All elements
  of $Y \in \CB$. However, $inf~ Y$ which is in $G$ is not an element of
  $\CB$. This implies that $\CB$ is not closed under the infimum operation.
Conversely,
  let $Y= \{H_1, H_2, \ldots, H_k\}$ be any subset of $\CB$ such that its
  infimum $G$ does not belong to $\CB$.
  Since $G$ is the infimum of $Y$, for any $i$, there exists $j\in \{1 \ldots k\}$
  such that $G[i] = H_j[i]$. Since $B(H_j)$ is true for all
  $j$, it follows that there exists a $G$ for which lattice-linearity
  does not hold.

(c) Follows from the equivalence of meet-closed predicates with lattice-linearity and that meet-closed predicates are closed
under conjunction. For a more direct proof, suppose that $\neg (B_1 \wedge B_2)$. This implies that one of the conjuncts is false and therefore from 
the lattice-linearity of that conjunct, a forbidden state exists.
\end{proof}

For the job scheduling example, we can define 
$B_j$ as $\C[j] \geq max ( t_j, max \{ \C[i] + t_j ~|~ i \in pre(j) \})$. Since $f_j(\C) = max ( t_j, max \{ \C[i] + t_j ~|~ i \in pre(j) \})$ is a monotone function,
it follows from Lemma \ref{lem:lin}(a) that $B_j$ is lattice-linear. The predicate $B_{jobs} \equiv \forall j: B_j$ is lattice-linear due to Lemma \ref{lem:lin}(c).
Also note that the problem of finding the minimum vector that satisfies $B_{jobs}$ is well-defined due to Lemma \ref{lem:lin}(b).

We now discuss detection of lattice-linear predicates which requires an additional assumption 
called
the {\em efficient advancement property} \cite{chase1998detection} --- there exists an efficient (polynomial
time) algorithm to determine the forbidden state. 
This property holds for all the problems considered in this paper.  Once we determine $j$ such that $forbidden(G,j,B)$, 
we also need to determine how to advance along index $j$.
To that end, we extend the definition of forbidden as follows.
\begin{definition}[$\alpha$-forbidden]
 Let $B$ be any boolean predicate on the lattice $\CC$ of all assignment vectors.
 For any $\C$, $j$ and positive real $\alpha > \C[j]$, we define $\mbox{forbidden}(\C,j, B, \alpha)$ iff
 $$  \forall H \in \CC:H \geq \C: (H[j] < \alpha) \Ra \neg B(H).  $$
\end{definition}

Given any lattice-linear predicate $B$, suppose $\neg B(\C)$. This means that $\C$ must be advanced on all
indices $j$ such that $\forbidden(\C,j,B)$.  We use a function $\alpha(\C,j, B)$ such that $\forbidden(\C, j, B, \alpha(\C,j, B))$ holds
whenever $\forbidden(\C,j,B)$ is true.  With the notion of $\alpha(\C, j, B)$, we have the algorithm $LLP$ shown in Fig. \ref{fig:alg-llp}.
The algorithm $LLP$ has two inputs --- the predicate $B$ and the top element of the lattice $T$. It returns the least vector $\C$ which is less than or equal to $T$
and satisfies $B$ (if it exists). Whenever $B$ is not true in the current vector $\C$, the algorithm advances on all forbidden indices $j$
in parallel. This simple parallel algorithm can be used to solve a large variety of combinatorial optimization problems
by instantiating different $\forbidden(\C,j,B)$ and $\alpha(\C,j,B)$.

\begin{figure}[htb]
{\small 
\fbox{\begin{minipage}[t]  {2.1in}
\begin{tabbing}
x\=xxxx\=xxxx\=xxxx\=xxxx\=xxxx\= \kill
\>vector {\bf function} getLeastFeasible($T$: vector, $B$: predicate)\\
\> {\bf var} $\C$: vector of reals initially $\forall i: \C[i] = 0$;\\
 \> {\bf while} $\exists j: \forbidden(\C,j,B)$ {\bf do}\\
    \>\>    {\bf for all} $j$ such that $\forbidden(\C,j,B)$  {\bf in parallel}:\\
    \>\>\> {\bf if} $(\alpha(\C,j,B) > T[j])$ then return null; \\
 \>   \>\>    {\bf else} $\C[j] := \alpha(\C,j,B)$;\\
      \> {\bf endwhile};\\
    \> {\bf return} $\C$; // the optimal solution
    \end{tabbing}
\end{minipage}
  }
  }
\caption{Parallel Algorithm  $LLP$ to find the minimum vector less than or equal to $T$ that satisfies $B$\label{fig:alg-llp}}
\end{figure}

\begin{theorem}\label{thm:mainllp}
Suppose there exists a fixed constant $\delta > 0$ such that $\alpha(\C, j, B) - \C[j] \geq \delta$ whenever $\forbidden(\C, j, B)$.
Then, the parallel algorithm $LLP$ finds the least vector $\C \leq T$ that satisfies $B$, if one exists.
\end{theorem}
\begin{proof}
Since $\C[j]$ increases by at least $\delta$ for at least one forbidden $j$ in every iteration of the {\em while} loop, the algorithm terminates
in at most $\sum_i \lceil T[i]/\delta \rceil$ number of steps.

We show that the algorithm maintains the invariant $(I1)$ that for all indices $j$, any vector $V$ such that $V[j]$ is less than $\C[j]$ cannot satisfy $B$.
Formally, the invariant $(I1)$ is
$$\forall j: (\forall V \in \CC: (V[j] < \C[j]) \Ra \neg B(V)).$$
 Initially, the
invariant holds trivially because $\C$ is initialized to $0$. Suppose $\forbidden(\C, j, B)$.
Then, we increase $\C[j]$ to $\alpha(\C, j, B)$. We need to show that this change maintains the invariant.
Pick any $V$ such that $V[j] < \alpha(\C, j, B)$. We now do a case analysis. If $V \geq \C$, then $\neg B(V)$ holds from the
definition of $\alpha(\C, j, B)$. Otherwise, there exists some $k$ such that $V[k] < \C[k]$. In this case $\neg B(V)$ holds due to $(I1)$.

We can now show Theorem \ref{thm:mainllp} using the invariant.
First, suppose that the algorithm $LLP$ terminates because $\alpha(\C, j, B) > T[j]$. In this case, there is no feasible vector in $\CC$ due to 
the invariant (because the predicate $B$ is false for all values of $\C[j]$). 
Now suppose that the algorithm terminates because there does not exist any $j$ such that $\forbidden(\C, j, B)$. This implies that 
$\C$ satisfies $B$ due to lattice-linearity of $B$. It is also the least vector that satisfies $B$ due to the invariant $(I1)$.
\end{proof}

For the job scheduling example, we get a parallel algorithm to find the minimum completion time
by using $\forbidden(\C,j,B_{jobs}) \equiv (\C[j] < t_j) \vee (\exists i \in pre(j): \C[j] < \C[i] + t_j)$,
and $\alpha(\C, j, B_{jobs}) = \max \{ t_j, \max \{ \C[i]+t_j  | i \in pre(j)\}\}$.

For the prefix sum example, we get a parallel algorithm by using 
$\forbidden(\C,j,B_{prefix}) \equiv  (\C[j] < \C[j-1] + A[j-1])$ and $\alpha(\C,j,B_{prefix}) = \C[j-1] + A[j-1]$ for all $j > 0$.

We now show, on account of Lemma \ref{lem:lin}(c),  that if we have a parallel algorithm for a problem, then we also have one
for the constrained version of that problem.
\begin{lemma}\label{lem:cons}
Let $LLP$ be the parallel algorithm to find the least vector $\C$ that satisfies $B_1$ if one exists. Then, $LLP$ can be adapted to
find the least vector $\C$ that satisfies $B_1 \wedge B_2$ for any lattice-linear predicate $B_2$.
\end{lemma}
\begin{proof}
The algorithm $LLP$ can be used with the following changes:
$\forbidden(\C,j ,B_1 \wedge B_2) \equiv \forbidden(\C,j ,B_1) \vee \forbidden(\C,j, B_2),$ and\\
$\alpha(\C, j, B_1 \wedge B_2) = \max \{\alpha(\C,  j, B_1), \alpha(\C, j,  B_2) \}$.\\
\end{proof}

For example, suppose that we want the minimum completion time of jobs with the additional lattice-linear constraint that $B_2(\C) \equiv (\C[1]=\C[2])$.
$B_2$ is lattice-linear with $\forbidden(\C, 1, B_2) \equiv (\C[1] < \C[2])$ and $\forbidden(\C, 2, B_2) \equiv (\C[2] < \C[1])$. By applying, Lemma \ref{lem:cons},
we get a parallel algorithm for the constrained version.

The job scheduling problem and the prefix sum problem are special cases of the following optimization problem: minimize $\C$ such that $\forall i: \C_i \geq f_i(\C)$ 
where each of $f_i$ is a monotone function on the lattice of reals. When the predicate $B$ is of the form $\forall i: \C_i \geq f_i(\C)$,  the problem is closely related to finding the least fixed point in a lattice using Knaster-Tarski's fixed point theorem \cite{tarski:fp,lassez:fp}. 

The algorithm in Fig. \ref{fig:alg-llp}
can be viewed as repeated iteration of a monotone function on the bottom element of a lattice similar to a constructive version 
of Knaster-Tarski's theorem.
In particular, let $\CC$ be any lattice with the bottom element $\bot$, the top element $\top$ and $f$ a monotone function from $\CC$ to $\RR_{\geq 0}$. Consider the set $F \subseteq \CC$ defined as $\{x ~|~x \geq  f(x) \}$.
If there exists $k$ such that $f^k(\bot) = f^{k+1}(\bot)$, then $x^* = f^k(\bot)$ is the minimum element in $\CC$ that satisfies the predicate $(x \geq f(x))$.
Our work differs from such earlier work in many respects. First, $B$ may not have the form $\forall i: \C_i \geq f_i(\C)$; instead we only require $B$ to 
be closed under meets. Second,  Knaster-Tarski's fixed point theorem (and many variants) requires the function to be from the lattice $\CC$ to itself.
In that case, the solution to the equation $x \geq f(x)$ always exists for a complete lattice because $\top \geq f(\top)$. We do not assume that the range of the function is the 
lattice itself. Therefore, there is no guarantee of the existence of the fixed point. Indeed, for the job scheduling example, if the prerequisites have a cycle
and weights are positive, then there is no solution and the algorithm $LLP$ returns null.
Third, the goal of this paper is to develop techniques to reach the fixed point with an efficient parallel algorithm and to
show that many standard and non-standard parallel algorithms for combinatorial optimization can be derived 
in this framework.


Note that the straightforward application of $LLP$ may not give the most time-efficient parallel algorithm. The efficiency of
the algorithm may depend upon $\alpha(\C,j, B)$ chosen for the predicate $B$. We will later show such optimizations for the shortest path algorithm.

\section{Constrained Stable Matching Problem}
\label{sec:csmp}
We now show that our technique is applicable when the search space is any distributive lattice rather than 
lattice of non-negative real vectors. We illustrate this by applying predicate detection to the stable matching problem.
Traditionally, a stable matching problem is modeled as a bipartite graph consisting of two sets of vertices for men and
women. We use a different model to exploit predicate detection techniques --- instead of the underlying set of vertices being men and women, it is the set 
of proposals $E$ that can be made by men to women. We call these proposals, {\em events},
which are executed by
$n$ processes corresponding to $n$ men denoted by $\{P_1 \ldots P_n\}$.  
An event is denoted by a tuple $(i,j)$ that corresponds to
the proposal made by man $i$  to woman $j$. Women are also numbered $1..n$ and
denoted by $\{w_1 \ldots w_n\}$.
We now impose a partial order $\ra_p$ on $E$  to model the order in which 
these proposals can be made. In the standard stable matching problem (SMP), every man $P_i$ has his preference list $mpref[i]$ such that
$mpref[i][k]$ gives the $k^{th}$ most preferred woman for $P_i$. We model $mpref$ using
$\ra_p$ order on the set of events. If $P_i$ prefers woman $j$ to woman $k$, then 
there is an edge from the event $(i,j)$ to the event $(i,k)$. As in SMP, we assume that 
every man gives a total order on all women. 
Each process makes proposals to women in the decreasing order of preferences (similar to the Gale-Shapley algorithm).

In the standard stable matching problem, there are no constraints
on the order of proposals made by different men, and $\ra_p$ can be visualized as a partial order $(E, \ra_p)$ with
$n$ disjoint chains.
In the constrained SMP, $\ra_p$  can relate proposals made by different men
and therefore $\ra_p$ forms a general poset $(E, \ra_p)$.
For example, the constraint that Peter's regret is less than or equal to John can be modeled by adding $\ra_p$ edges
as follows. For any regret $r$, we add an $\ra_p$ edge from the proposal by John with regret $r$ to the
proposal by Peter with regret $r$.
The constraint that Peter and John have the same regret can be modeled as a conjunction of Peter's regret being less than that 
of John and vice-versa.
We draw $\ra_p$ edges in solid (blue) edges as shown in Fig. \ref{fig:csmp-model}.




Let $G \subseteq E$ denote the  global state of the system. A global state $G$ is simply the subset of events executed in the computation
such that it preserves the order of events within each $P_i$. 
Since all events executed by a process $P_i$ are totally ordered,
 it is sufficient to record the number of events executed 
by each process in a global state. 
Let $G[i]$ be the number of proposal made by $P_i$. 
Initially, $G[i]$ is $0$ for all men.
If $P_i$ has made $G[i]>0$ proposals, then 
$mpref[i][G[i]]$ gives the identity of the woman last proposed by $P_i$. 
We let $event(i, G[i])$ denote the event in which $P_i$ makes a proposal to $mpref[i][G[i]]$.
We also use $succ(event(i, G[i]))$ to denote the next proposal made by $P_i$, if any.


For the constrained SMP, we have $\ra_p$ edges that relate proposals of different processes.
The second graph in Fig. \ref{fig:csmp-model} shows an example of using $\ra_p$ edges in the constrained SMP.
For this problem, we work with {\em consistent global states} (or order ideals \cite{davey}).
\begin{definition}[Consistent Global State]
A global state $G \subseteq E$ is {\em consistent} if
$ \forall e,f \in E: (e \ra_p f) \wedge (f \in G) \Ra (e \in G).$
\end{definition}
In the context of constrained SMP, it is easy to verify that $G$ is consistent iff 
for all $j$, there does not exist $i$ such that $succ(event(j, G[j])) \ra_p event(i, G[i])$
(otherwise, by letting $e = succ(event(j,G[j])$ and $f = event(i, G[i])$, we get a contradiction to consistency).


It is well known that the set of all consistent global states of a finite poset forms a finite
distributive lattice \cite{davey}. We use the lattice of all consistent global states as $\CC$ for
the predicate detection.

In the standard SMP, women's preferences
are specified by preference lists $wpref$ such that $wpref[i][k]$ gives the $k^{th}$ most preferred man for woman $i$.
It is also convenient to define $rank$ such that $rank[i][j]$ gives the choice number $k$ for which $wpref[i][k]$ equals $j$, i.e.,
$wpref[i][k] = j$ iff $rank[i][j] = k$.
We model these preferences using edges on the computation graph as follows. If an event $e$ 
corresponds to a proposal by $P_i$ to woman $q$ and she prefers $P_j$, then we add a dashed (green) edge
from $e$ to the event $f$ that corresponds to $P_j$ proposing to woman $q$.
The set $E$ along with the dashed edges also forms a partial order $(E, \ra_w)$ where 
$e \ra_w f$ iff both proposals are to the same woman and that woman prefers the proposal
$f$ to $e$. 
It is important to note that the definition of a consistent global state only refers to $\ra_p$.
The relation $\ra_w$ is used only to define the feasible predicate later.
With $((E, \ra_p), \ra_w)$ we can model any SMP specified using $mpref$ and $wpref$ as shown in Fig. \ref{fig:SMP}.
The first graph in Figure \ref{fig:csmp-model} gives an example of a standard SMP.
To avoid cluttering the
figure, we have shown preferences of all men but preferences of only the woman $w_1$.
The woman $w_1$ has highest preference for $P_4$, followed by $P_1$, $P_3$ and $P_2$.
Note that an arrow in $\ra_p$ for the same man goes from 
the more preferred  woman to the less preferred woman, but an arrow in $\ra_w$ for the same woman
goes from the less preferred man to the more preferred man.

The second graph in Fig \ref{fig:csmp-model} gives a constrained SMP. Since both $\ra_p$ and $\ra_w$ are
transitive relations, we draw only the transitively reduced diagrams.

\begin{small}
\begin{figure}
\begin{tabular}{l | l l l l |   l l   l | l l l l |}
$mpref$ & & & &  & & & $wpref$ \\
$P_1$  &   $w_4$ & $w_1$ & $w_2$ & $w_3$  &  & & $w_1$ & $P_4$ & $P_1$ & $P_3$ & $P_2$\\
$P_2$   &   $w_2$ & $w_3$ & $w_1$ & $w_4$ &  & & $w_2$ & $P_1$ & $P_4$ & $P_2$ & $P_3$\\
$P_3$   &   $w_3$ & $w_1$ & $w_4$ & $w_2$ &  & & $w_3$ & $P_1$ & $P_2$ & $P_4$ & $P_3$\\
$P_4$   &   $w_2$ & $w_4$ & $w_3$ & $w_1$ &  & & $w_4$ & $P_3$ & $P_1$ & $P_4$ & $P_2$\\
\end{tabular}
\caption{\label{fig:SMP}  Stable Matching Problem with men preference list ($mpref$) and
women preference list ($wpref$).}
\vspace*{-0.1in}
\end{figure}
\end{small}

\begin{figure}[htbp]
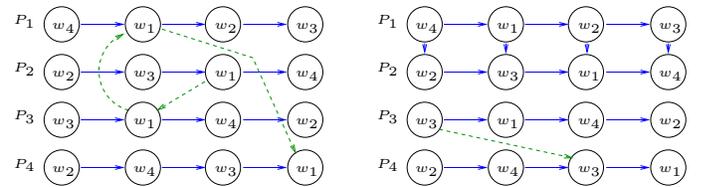
\caption{\label{fig:csmp-model} 
The first graph models the standard SMP problem. Men's preferences are shown in blue solid edges. Preferences of women 1 and 2  are shown in dashed green edges. The second
 constrained SMP Graph corresponds to constraint that the {\em regret} for $P_2$ is less than or equal to that of $P_1$.
It also shows the preference of $w_3$ of $P_4$ over $P_3$.}
\end{figure}


The above discussion motivates the following definition.
\begin{definition}[Constrained SMP Graph]
Let $E = \{(i,j) | i \in [1..n] \mbox{ and } j \in [1..n] \}$.
A Constrained SMP Graph $((E, \ra_p), \ra_w)$ is a directed graph on $E$ with two sets of edges $\ra_p$ and $\ra_w$ with the following
properties: (1) $(E, \ra_p)$ is a poset such that the set $P_i = \{ (i,j) | j \in [1..n] \}$ is a chain for all $i$, and
(2) $(E, \ra_w)$ is a poset such that the set $Q_j = \{ (i,j) | i \in [1..n] \}$ is a chain for all $j$  and there is no $\ra_w$ edge
between proposals to different women, i.e.,
for all $i,j,k,l: (i,j) \ra_w (k,l) \Ra (j=l)$.
\end{definition}

Given a global state $G$, we define the {\em frontier} of $G$ as the set of maximal events executed by 
any process. The frontier includes only the last event executed by $P_i$ (if any). Formally,
$\frontier(G) = \{ e \in G ~|~ \forall f \in G$ such that  $f \neq e$, $f$ and $e$ are executed by $P_i$ implies $f \ra_p e$ \}.
We call the events in $G$ that are not in $\frontier(G)$ as pre-frontier events.
For example, suppose that $G = \{(P_1, w_4), (P_1, w_1), (P_2, w_2), (P_3, w_3), (P_3, w_1), (P_4, w_2)\}$ in Fig. \ref{fig:csmp-model}.
Then, $$\frontier(G) = \{ (P_1, w_1), (P_2, w_2), (P_3, w_1), (P_4, w_2) \}$$ and 
$\mbox{pre-frontier}(G) = \{ (P_1, w_4), (P_3, w_3) \}$.


We now define the feasible predicate on  global states as follows.
\begin{definition}[feasibility for marriage]
A global state $G$ is  feasible for marriage iff (1) every man has made at least one proposal
(2)  $G$ is a consistent global state, and (3) there is no green edge ($\ra_w$) from a frontier event to any event of $G$ (frontier or pre-frontier).
Formally, $B_{marriage}(G) \equiv$\\
$ (\forall i: G[i] > 0) \wedge consistent(G) \wedge  (\forall e \in \frontier(G), \forall g \in G: \neg (e \ra_w g).$

\end{definition}
A green edge from a frontier event $(i,j)$ to another frontier event $(k,l)$ would imply that $j$ is equal to $l$
and both men $i$ and $k$ are assigned the same woman which violates the matching property.
For example, in Fig. \ref{fig:csmp-model}, $(P_3,w_3)$ and $(P_4, w_3)$ cannot both be frontier events of an feasible global state.
A green edge from a frontier event $(i,j)$ to a pre-frontier event $(k,l)$ in a global state $G$ implies that woman $j$ prefers man $k$ to her partner $i$ and man $k$ prefers $j$ to his partner in $G$.
For example, in Fig. \ref{fig:csmp-model}, $(P_3,w_3)$ and $(P_4, w_1)$ cannot both be frontier events of an feasible global state because there is
a green edge from $(P_3, w_3)$ to a pre-frontier event $(P_4, w_3)$.

It is easy to verify that the problem of finding a stable matching is the same as finding a global state that 
satisfies the predicate $B_{marriage}$ which is defined purely in graph-theoretic terms on the
constrained SMP graph.
The next task is to show that $B_{marriage}$ is lattice-linear.

\begin{theorem}\label{lem:CSMP-forbid}
For any global state $G$ that is not a constrained stable matching,
there exists a $j$ such that $\forbidden(G,j,B_{marriage})$.
\end{theorem}
\begin{proof}
First suppose that some process $j$ has not made any proposal (i.e., $G[j] = 0$), then clearly
unless process $j$ is advanced, the predicate $B_{marriage}$ cannot hold. Now suppose that $G$ is not 
a consistent global state. This means that there exists an event $e$, say on $P_j$, such that $e \ra_p f$ for
some event $f \in G$, but $e \not \in G$. Again, unless $G$ is advanced on $P_j$, the global state cannot become consistent
and hence $j$ is forbidden in $G$.
Finally, suppose that there exists a green edge from a frontier event of $G$, $G[j]$ to another
event $g \in G$. Consider $H$ such that $G \subseteq H$ and $G[j] = H[j]$. We now have that $H[j]$ is a frontier event
and it has a green edge to $g$. Since $g$ is an event in $G$, it is also included 
in all global states greater than $G$. Hence, $H$ is also not feasible. Therefore, $\forbidden(G, j, B)$.
\end{proof}

We now apply the detection of lattice-linear global predicates for the constrained stable matching. 

\begin{figure}[htb]\begin{center}
\fbox{\begin{minipage}  {\textwidth}\sf
\begin{tabbing}
xx\=xxxx\=xxxx\=xxxx\=xxxx\=xxxx\= \kill
{\bf Algorithm Constrained-Stable-Matching}:\\ 
Use Algorithm $LLP$ where\\
$T$ = $(n,n,...,n)$; //maximum number of proposals at $P_i$\\
$z = mpref[j][G[j]]$; //current woman assigned to man $j$\\
\\
$\forbidden(G, j, B_{marriage}) \equiv   (G[j] = 0)$ \\
 \>  $\vee (\exists i: \exists k \leq G[i]: (z = mpref[i][k])$\\ 
 \> \> $\wedge (rank[z][i] < rank[z][j]))$\\
 \> $\vee (\exists i: succ(event(j, G[j])) \ra_p event(i, G[i]]))$ \\
\\
$\alpha(G,j,B_{marriage}) = (G[j]+1)$;
\end{tabbing}
\end{minipage}
} 
\end{center}
\caption{An efficient algorithm to find the man-optimal constrained stable matching less than or equal to $T$ \label{fig:alg}}
\end{figure}

%
%

The algorithm to find the man-optimal constrained stable marriage is shown in Fig. \ref{fig:alg}.
From the proof of Theorem \ref{lem:CSMP-forbid}, we get the
following implementation of $\forbidden(G, j, B_{marriage})$ in Fig. \ref{fig:alg}.
The first disjunct in the function $\forbidden$ holds when the assignment for $G$ is null with respect to $j$.
The second disjunct holds when the woman $z$ assigned to man $j$ is such that there exists a man $i$
who is either (1) currently assigned to $z$ and woman $z$ prefers man $i$ , or (2) is currently assigned
to another woman but he prefers $z$ to the current assignment. The first case holds when $k=G[i]$ and the
second case holds when $k < G[i]$.
The first case is equivalent to checking is a green edge exists from $(j, z)$ to a frontier event.
The second case is equivalent to checking if a green edge exists to a pre-frontier event.
The third disjunct checks that the assignment for $G$ satisfies all external constraints with respect to $j$.

Our algorithm generalizes the Gale-Shapley algorithm in that it allows specification of external constraints.
It is easy to see that  the third disjunct ($\exists i:succ(event(j, G[j]))  \ra_p event(i,G[i]))$ is not necessary when there are no
external constraints. It can also be shown that the second disjunct can be 
simplified to  $(\exists i: (z = mpref[i][G[i]]) \wedge (rank[z][i] < rank[z][j]))$ in the absence of external constraints.
On removing the third disjunct and simplifying the second disjunct, we get the Gale-Shapley algorithm as a special case of our algorithm. 
In every iteration of $LLP$, we need to find $j$ such that $\forbidden(G,j,B)$ holds.
By keeping a list of $j$ such that $G[j]$ is equal to $0$, finding $j$ that satisfies the first disjunct
can be done in $O(1)$ time. 
By keeping the rank of the most preferred man that has proposed to any woman with the woman, and the list of men that have been rejected,
finding $j$ that satisfies the second disjunct can also be done in $O(1)$ time.
For the standard matching problem, these two disjuncts are sufficient and we get the overall time complexity of $O(n^2)$ which is identical to
that of the Gale-Shapley algorithm. 


We note here that existence of a stable matching via Tarski's theorem is shown by Adachi \cite{adachi2000characterization}.
Our interest is not in the existence (and occasionally there may not be any solution to the constrained stable marriage problem), but
the computation of the constrained stable marriage whenever it exists.

\section{Constrained Single Source Shortest Path Algorithm}
\label{sec:csssp}

In this section, we first model the traditional shortest path problem and later generalize it to incorporate constraints.
Consider a weighted directed graph with $n$ vertices numbered $0$ to $n-1$.  We assume that all edge weights are {\em strictly positive}.
We are required to find the minimum cost of a path from a distinguished {\em source} vertex $v_0$ 
to all other vertices where the cost of a path is defined as the sum of edge weights along that path.
For any vertex $v$, let $pre(v)$ be the set of vertices $u$ such that  $(u,v)$ is an edge in the graph.
To avoid trivialities, assume that every vertex  $v$ (except possibly the source vertex $v_0$) has nonempty $pre(v)$.

As the first step of the predicate detection algorithm, we define the lattice for the search space. We assign to each vertex $v_i$, $\C[i] \in \RR_{\geq 0}$ with the interpretation that
$\C[i]$ is the cost of reaching vertex $v_i$. 
We call $\C$, the {\em assignment} vector.
The invariant maintained by our algorithm is:
for all $i$, the cost of any path from $v_0$ to $v_i$ is greater than or equal to $\C[i]$.

\begin{figure}[htbp]
\begin{center}
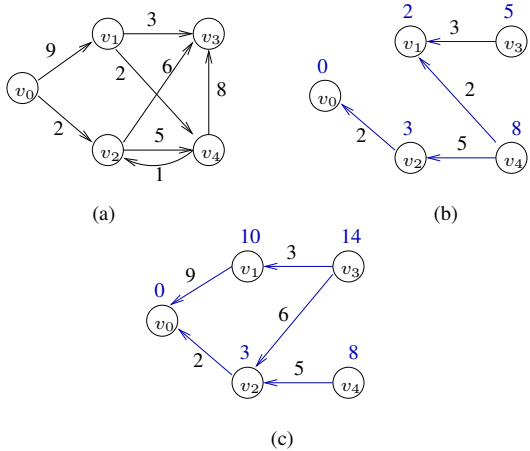
\caption{\label{fig:myGraph} (a) A Weighted Directed Graph (b) The parent structure for $\C=(0,2,3,5,8)$ (c) The parent structure for $\C=(0,10,3,14,8)$.
Since every non-source node has at least one parent, $\C$ is feasible. }
\end{center}
\vspace*{-0.2in}
\end{figure}

The second step in the predicate detection algorithm is to define an appropriate {\em feasibility} predicate.
The vector $\C$ only gives the lower bound on the cost of a path and there may not be any 
path to vertex $v_i$ with cost $\C[i]$. To capture that an assignment is feasible, we define
{\em feasibility} which requires the notion of a {\em parent}. We say that $v_i$ is a parent of $v_j$ in $\C$ (denoted by the predicate $parent(j, i, \C)$)  iff
there is a direct edge from $v_i$ to $v_j$ and $\C[j]$ is at least $(\C[i]+w[i,j])$, i.e.,
$(i \in pre(j)) \wedge (\C[j] \geq \C[i] + w[i,j])$. 
In Fig. \ref{fig:myGraph}, let $\C$ be the vector $(0, 2, 3, 5, 8)$. Then, $v_0$ is a parent of $v_2$ because $\C[2]$ is
greater than $\C[0]$ plus $w[0,2] (i.e., 3 \geq 0 + 2)$. Similarly, $v_1$ is a parent of $v_4$ because $\C[4] \geq \C[1] + 2$.
A node may have multiple parents. The node $v_2$ is also a parent of $v_4$ because $\C[4] \geq \C[2] + 5$.
%
Since $w[i,j]$ are strictly positive, there cannot be a cycle in the parent relation.
Now, feasibility can be defined as follows.

\begin{definition}[feasible for paths]
An assignment $\C$ is {\em feasible for paths } 
 iff every node except the source node has a parent. Formally,
$B_{path}(\C) \equiv \forall j  \neq 0:  (\exists i: parent(j, i, \C)).$

\end{definition}
Hence, an assignment $\C$ is feasible iff one can go from any non-source node to the source node by following any of the
parent edges. In Fig. \ref{fig:myGraph}, the vector $\C = (0, 10, 3, 14, 8)$ is feasible because every non-source node has at least one parent.
One can go from $v_3$ to $v_0$ either via $v_1$ or $v_2$.

As the third step of our method, we now show that feasibility satisfies lattice-linearity.

\begin{lemma}\label{lem:CSSP-forbid}
For any assignment vector $\C$ that is not feasible,
$\exists j: \forbidden(\C,j, B_{path}(\C))$.
\end{lemma}
\begin{proof}
Suppose $\C$ is not feasible. Then, 
there exists $j \neq v_0$ such that $v_j$ does not have a parent, i.e., 
$ \forall i \in pre(j): \C[j] < \C[i] + w[i,j]. $ Of all nodes that do not have any parent, let
$j$ be such that  it has the least value of $\C[j]$.
If there are multiple $j$ with the same value of $\C[j]$, then we choose any $j$.
We show that $\forbidden(\C, j, B_{path}(\C))$ holds. 
Pick any $H$ such that $H \geq \C$. Since for any $i \in pre(j)$, $H[i] \geq \C[i]$, 
$\C[j] < \C[i] + w[i,j]$ implies that $\C[j] < H[i]+w[i,j]$. Therefore, whenever $H[j] = \C[j]$,
$v_j$ does not have a parent.
%
%
%
\end{proof}

{\bf Remark.} Since $B_{path}$ is a lattice-linear predicate, it follows from Lemma \ref{lem:lin}(c), that the set of feasible assignment vectors
are closed under meets (the component-wise min operation). Unlike the constrained stable matching problem, the set of feasible assignment vectors are not closed under the join operation.
In Fig. \ref{fig:myGraph}, the vectors $(0, 10, 3, 14, 8)$ and $(0, 9, 10, 12, 11)$ are feasible, but their join $(0, 10, 10, 14, 11)$ is not feasible.

\begin{lemma}\label{lem:adv-alpha}
Suppose $\neg B_{path}(\C)$.
Then, $(\forall i: \neg parent(j, i, \C)) \Ra \forbidden(\C, j, B_{path}, \alpha(\C, j) )$ for any $j \neq 0$, where
$ \alpha(\C,j, B_{path}) =  \min \{ \C[i] + w[i,j] ~|~ i \in pre(j)\} .$
\end{lemma}
\begin{proof}
It is sufficient to observe that $\alpha(\C,j)$ is the minimum amount that $\C[j]$ must increase to have a parent.
\end{proof}

\begin{figure}[htb]
  {\small 
\fbox{\begin{minipage}[t]  {\textwidth}
\begin{tabbing}
xx\=xxxx\=xxxx\=xxxx\=xxxx\=xxxx\= \kill
 Use Algorithm $LLP$ with 
    \\
    \> $T = (M, M, \ldots, M)$, where $M = n * \max_{i,j: i \in pre(j)} w[i,j] $\\
\>   $parent(j, i, \C) \equiv (i \in pre(j)) \wedge (\C[j] \geq \C[i] + w[i,j])$\\
\>      $\forbidden(\C,j, B_{path}) \equiv \forall i: \neg parent(j, i, \C)$\\

\>     $\alpha(\C,j,  B_{path}) =  \min \{ \C[i] + w[i,j]  ~|~ i \in pre(j) \} $
   
\end{tabbing}
\end{minipage}
} 
}
\caption{{\bf Algorithm $ShortestPath_a$}
to find the minimum cost assignment vector less than or equal to $T$ \label{fig:alg-alpha}}
\end{figure}

We now have our first algorithm for computing the cost of the shortest path from a source vertex.
The algorithm called $ShortestPath_a$ (Lattice-Linear Shortest Path) shown in Fig. \ref{fig:alg-alpha} finds the minimum cost assignment vector $\C$ less than or equal to $T$ that is feasible.
We note that the algorithm returns null if any vertex in the graph is not reachable from the source vertex because
there is no assignment vector less than or equal to $T$ that is feasible. If one wants the algorithm to find the shortest path
to a specific target vertex, then the algorithm can be modified to delete the unreachable vertex from the graph (and the corresponding
component from $\C$ array) and then continue until either the desired target vertex is deleted or there exists no forbidden vertex.
If there is no forbidden vertex in $\C$, then any component $j$ of $\C$ has the cost of the shortest path from the source vertex to $j$.

For an unweighted graph (i.e., each edge has weight equal to $1$), the above parallel algorithm requires
time equal to the distance of the farthest node from the root. However, for a weighted graph, it may take 
non-polynomial time in the number of nodes because the advancement along a forbidden process may be small.

We now give an alternative feasible predicate that results in an algorithm that takes bigger steps.
%
%
We first define a node $j$ to be {\em fixed} in $\C$ if either it is the source node or it has a parent
that is a fixed node, i.e.,
$ fixed(j,\C) \equiv (j=0) \vee (\exists i: parent(j, i, \C) \wedge fixed(i,\C)).$

Observe that node $v_0$ is always fixed. Any node $v_j$ such that one can reach from $v_j$ to $v_0$ using
parent relation is also fixed.
We now define another feasible predicate called $B_{rooted}$, as
$$B_{rooted}(\C) \equiv  \forall j: fixed(j,\C).$$

Even though it may first seem that the predicate $B_{rooted}$ is strictly stronger than $B_{path}$, the following Lemma shows otherwise.
\begin{lemma}
$B_{path}(\C)$
iff  $B_{rooted}(\C)$.
\end{lemma}
\begin{proof}
If $\C$ satisfies $B_{rooted}$, then every node other than $v_0$ has at least one parent by definition of $fixed$, hence $B_{path}(\C)$.
Conversely, suppose that every node except $v_0$ has a parent. Since parent edges cannot form a cycle, by following the parent
edges, we can go from any node to $v_0$.
\end{proof}

It follows that the predicate $B_{rooted}$ is also lattice-linear. An advantage of the predicate  $B_{rooted}$ is that
it allows us to define a different $\alpha$ for advancement whenever the assignment vector $\C$ is not $B_{rooted}$.
Suppose that $\neg B_{rooted}(\C)$. Then, the following threshold $\beta(\C)$ is well-defined whenever
the set of edges from the fixed vertices to non-fixed vertices is nonempty.
$$
\beta(\C) = \min_{(i,j): i \in pre(j) } \{ \C[i] + w[i,j]   ~|~  fixed(i,\C),  \neg fixed(j,\C) \}.
$$
If the set of such edges is empty then no non-fixed vertex is reachable from the source.

We now have the following result in advancement of $\C$.
\begin{lemma}\label{lem:adv-beta}
Suppose $\neg B_{rooted}(\C)$.
Then, $\neg fixed(j, \C) \Ra \forbidden(\C, j, B_{rooted}, \beta(\C) )$.
\end{lemma}
\begin{proof}
Consider any assignment vector $H$ such that $H \geq \C$ and $H[j] <  \beta(\C)$.
We show that $H$ is not $B_{rooted}$. In particular, we show that $j$ is not fixed in $H$.
Suppose $j$ is fixed in $H$. This implies that there is a path $W$ from $v_0$ to $v_j$ such that all nodes in that path are
fixed. Let the path be the sequence of vertices $w_0, w_1, \ldots w_{m-1}$, where $w_0=v_0$ and $w_{m-1} = v_j$.
Let $w_l = v_{k}$ be the first node in the path that is not fixed in $\C$. Such a node exists because $w_{m-1}$ is not fixed in $\C$.
Since $w_0$ is fixed, we know that $1 \leq l \leq m-1$.
The predecessor of $w_l$ in that path, $w_{l-1}$ is well-defined because $l \geq 1$.
Let $w_{l-1} = v_{i}$.

We show that $H[k] \geq  \beta(\C)$ which contradicts $H[j] <  \beta(\C)$ because $H[k] \leq H[j]$ as the cost can only
increase going from $k$ to $j$ along the path $W$. We have $H[k] \geq H[i] + w[i,k]$ because $i$ is a parent of $k$ in $H$. Therefore,
$H[k] \geq \C[i] + w[i,k]$ because $H[i] \geq \C[i]$. Since $i$ is fixed in $\C$ and $k$ is not fixed in $\C$, from the definition 
of $\beta(\C)$, we get that $\beta(\C) \leq \C[i] + w[i,k]$. Hence, $H[k] \geq \beta(\C)$.
\end{proof}

By combining the advancement Lemma \ref{lem:adv-beta} with Lemma \ref{lem:adv-alpha}, we get the algorithm $ShortestPath_b$ shown in Fig. \ref{fig:alg-beta}.
In this algorithm, in every iteration of the while loop we find an edge going from a fixed node $i$ to a non-fixed node $j$ that minimizes
$\C[i] + w[i,j]$. All nodes are advanced to $\alpha(\C, j)$  that combines $\beta(\C)$ with $\alpha(\C, j)$ of the algorithm $ShortestPath_a$ as shown in Fig. \ref{fig:alg-beta}.
Note that if a node is fixed, its parent is fixed and therefore any algorithm that advances $\C[j]$ only for non-fixed nodes $j$
maintain that once a node becomes fixed it stays fixed.

In the graph of Fig. \ref{fig:myGraph}, $\C$ is initially $(0,0,0,0,0)$. 
After the first iteration, $\beta(\C) = 2$ and $\alpha(1, \C) = 9, \alpha(2, \C) = 2, \alpha(3, \C) = 3, \alpha(4, \C) = 2$. Therefore, $\C$ becomes $(0,9,2,3,2)$. At this point, nodes
$v_0$, $v_1$ and $v_2$ are fixed. In the second iteration, $\beta(\C) = 7$, $\alpha(3, \C) = \max(7, \min(12, 8, 10)) = 8$ and $\alpha(4, \C) =
\max(7, \min(11, 7)) = 7$. Hence, $\C = (0, 9, 2, 8, 7)$. At this point, all nodes are fixed and the algorithm terminates.


We first note that by removing certain steps, we get Dijkstra's algorithm from the algorithm $ShortestPath_b$. It is clear that
the algorithm stays correct if $\alpha(\C, j)$ uses just $\beta(\C)$ instead of $\max \{ \beta(\C),  \min \{ \C[i] + w[i,j]  ~|~ i \in pre(j) \} \}$.
Secondly, the algorithm stays correct if we advance $\C$ only on the node $j$ such that $(i, j) \in E'$ minimizes $\C[i] + w[i,j]$.
Finally, to determine such a node and $\beta(\C)$, it is sufficient to maintain a min-heap of
all non-fixed nodes $j$, along with the label that
equals $\min_{i \in pre(j), fixed(i, \C)}  \C[i] + w[i,j]$. 
On making all these changes to $ShortestPath_b$, we get Dijkstra's algorithm
(modified to run with a heap).

\remove{
It is illustrative to compare the algorithm $ShortestPath_b$ with Dijkstra's algorithm. In Dijkstra's algorithm, the nodes become fixed in the order
of the cost of the shortest path to them. In the proposed algorithm, a node may become fixed even when nodes with shorter cost have not been discovered.
In Fig. \ref{fig:myGraph}, node $v_1$ becomes fixed earlier than nodes $v_3$ and $v_4$.
This feature is especially useful when we are interested in finding the shortest path to a single destination and that
destination becomes fixed sooner than it would have been in Dijkstra's algorithm.
Dijkstra's algorithm maintains a distance vector $dist$ such that it is always feasible, i.e., 
for any vertex $v$ there exists a path from source to $v$ with cost less than or equal to $dist[v]$.
We maintain the invariant that the cost of the shortest path from source to $v$
is guaranteed to be at least $\C[v]$. Therefore, in Dijkstra's algorithm,
$dist[v]$ is initialized to $\infty$ whereas we initialize $\C$ to $0$. 
Dijkstra's algorithm and indeed many algorithms for combinatorial optimization, such as simplex, start
with a feasible solution and march towards the optimal solution, our algorithm starts with an extremal point in search space (even if it is infeasible)
and marches towards feasibility. Also note that in Dijkstra's algorithm, $dist[v]$ can only decrease during execution. 
In our algorithm, $\C$ can only increase with execution.
}
The $ShortestPath_b$ algorithm and indeed all the algorithms in this paper have a single variable $\C$. All other predicates and functions are defined
using this variable. 


\begin{figure}[htb]
{\small 
\fbox{\begin{minipage}[t]  {\textwidth}
\begin{tabbing}
xxxx\=xxxx\=xxxx\=xxxx\=xxxx\=xxxx\= \kill
Use Algorithm $LLP$ with \\
  \> $T = (M, M, \ldots, M)$, where $M = n * \max_{i,j} w[i,j] $\\
 \>       $parent(j, i, \C) \equiv (i \in pre(j)) \wedge (\C[j] \geq \C[i] + w[i,j])$\\
 \>   $ fixed(j,\C) \equiv (j=0) \vee (\exists i: parent(j, i, \C) \wedge fixed(i,\C)) $\\
 \>    $\forbidden(\C,j) \equiv \neg fixed(j, \C)$\\

\>  $E'$ := \{ $(i,k) ~|~i \in pre(k) \wedge fixed(i,\C) \wedge \neg fixed(k,\C)\} $;\\
\>      $\beta(\C) =  \min \{ \C[i] + w[i,j] ~|~ (i,j) \in E' \}  $\\
 
 \>      $\alpha(\C,j) = \max \{ \beta(\C),  \min \{ \C[i] + w[i,j]  ~|~ i \in pre(j) \} \}$

\end{tabbing}
\end{minipage}
} 
}
   
\caption{ {\bf Algorithm $ShortestPath_b$} to find the minimum cost assignment vector less than or equal to $T$. \label{fig:alg-beta}}
\end{figure}

We now consider the generalization of the shortest path algorithm with constraints. We assume that all constraints specified are lattice-linear.
For example, consider the following constraints:
\begin{itemize}
\item
Find the minimum cost vector such that cost of vertex $i$ is at most cost of vertex $j$. The predicate $B \equiv \C[j] \geq \C[i]$
is easily seen to be lattice-linear. If any cost vector $\C$ violates $B$, then the component $j$ is forbidden (with $\alpha(\C, j)$ equal to $\C[i]$).
\item
The predicate $(\C[i] = \C[j])$ is lattice-linear because it can be written as a conjunction of two lattice-linear predicates
$(\C[i] \geq \C[j])$ and $(\C[j] \geq \C[i])$.
\item
The predicate $B \equiv (\C[i] \geq k) \Ra (\C[j] \geq m)$ is also lattice-linear. If any cost vector violates $B$, then
we have $(\C[i] \geq k) \wedge (\C[j] < m)$. In this case, the component $j$ is forbidden with $\alpha(\C, j)$ equal to $m$.
\item
The predicate $B \equiv (\forall i: \C[i] \geq F[i])$  for any fixed vector $F$ is lattice-linear. It is sufficient to show that $(\C[i] \geq F[i])$ is lattice-linear.
If the predicate $(\C[i] \geq F[i])$ is false, then the component $i$ is forbidden with $\alpha(\C, j)$ equal to $F[i]$.
\end{itemize}

Again, from Lemma \ref{lem:cons}, the algorithm $LLP$ can be used to solve the constrained shortest path algorithm by 
combining $\forbidden$ and $\alpha$ for constraints with $B_{rooted}$.

An example of a predicate that is not lattice-linear is $B \equiv \C[i] + \C[j] \geq k$. If the predicate is false for $\C$, then 
we have $\C[i] + \C[j] < k$. However, neither $i$ nor $j$ may be forbidden. The component $i$ is not forbidden because if $\C[i]$ is fixed but $\C[j]$ is increased, the
 predicate $B$ can become true. Similarly, $j$ is also not forbidden.

\section{Constrained Market Clearing Price}
\label{sec:cmcp}

\newcommand{\B}{U}
In this section, we apply our technique to the problem of finding a market clearing price with constraints.
%
Let $I$ be a set of indivisible $n$ items, and 
$\B$, a set of $n$ bidders. 
Every item $i\in I$ is given a {\em valuation} 
$v_{b,i}$ by each bidder $b \in \B$. The valuation of any item $i$
is a 
number between $0$ and $T[i]$.
Each item $i$ is given a price $\C[i]$ which is also a 
number between $0$ and $T[i]$.
We are assuming integral costs for simplicity --- the algorithm is easily extensible to real costs.
%

Given a price vector $\C$, we define the bipartite graph $(I, \B, E(\C))$ as
\[ (j,b) \in E(\C) \equiv
\forall i: (v_{b,j} - \C[j]) \geq (v_{b,i} - \C[i]). \]
Informally, an edge exists between item $i$ and bidder $b$ if the payoff for the bidder (the bid minus the price) is maximized with that item.
Given any set $\B' \subseteq \B$, let $N(\B', \C)$ denote all the items that are adjacent to
the vertices in $\B'$ in the graph $(I, \B, E(\C))$.
A price vector $\C$ is a {\em market clearing price}, denoted by $B_{clearingPrice}(\C)$ if
the bipartite graph $(I, \B, E(\C))$ has a perfect matching. 
We now generalize the problem of finding a market clearing price to that of finding a constrained 
market clearing price. 
Given any set of valuations, and a boolean predicate $B$ that is a conjunction of lattice-linear constraints, a
price vector $\C$ is a {\em constrained market clearing price}, denoted by $constrainedClearing(\C)$
iff $clearing(\C) \wedge B(\C)$. From Lemma \ref{lem:cons}, it is sufficient to give an algorithm for $clearing(\C)$.

%
%
%

We now claim that

\begin{lemma}
The predicate $B_{clearingPrice}(\C)$ is a lattice-linear predicate on the lattice of price vectors.
\end{lemma}
\begin{proof}
%
From $\neg B_{clearingPrice}(\C)$, we get that
$(I, \B, E(\C))$ does not have a perfect matching.
From Hall's theorem, there exists a {\em minimal} set of over-demanded items $J$, i.e., 
there exists a set of bidders $\B'$ with item set $J=N(\B', \C)$ such that 
the size of $J$ is smaller than $\B'$ and no proper subset of $J$ is over-demanded. It can be shown that 
any item $j \in J$ satisfies
$\forbidden(j, \C, B_{clearingPrice})$ (from the proof of Theorem 1 in  \cite{demange1986multi}). 
%

\end{proof}

\begin{figure}[htb]
  {\small 
\fbox{\begin{minipage}[t]  {\textwidth}
\begin{tabbing}
xx\=xxxx\=xxxx\=xxxx\=xxxx\=xxxx\= \kill
 Use Algorithm $LLP$ with 
    \\
     \> $\forall i: T[i] =  \max_{b} v_{b,i} $\\
\>   $E(\C) =  \{ (i,b) ~|~ \forall j:  (v_{b,i} - \C[i]) \geq (v_{b,j} - \C[j])\} $\\
\> $overDemanded(J, \C) \equiv \exists U' \subseteq U: (J = N(U', \C))$\\
\> \>  $\wedge (|J|  < |U'|)$\\
 \> $minimalOverDemanded(J, \C)   \equiv  overDemanded(J, \C)$\\
 \> \>  $\wedge \forall J' \subseteq J: overDemanded(J', \C) \Rightarrow (J' = J)$\\
\>   $\forbidden(\C,j, B_{clearingPrice}) \equiv$\\
\> \> $ \exists J: minimalOverDemanded(J, \C) \wedge (j \in J) $\\
\> $\alpha(G,j,B_{clearingPrice}) = (G[j]+1)$;
%
   
\end{tabbing}
\end{minipage}
} 
}
\caption{{\bf Algorithm ConstrainedMarketClearingPrice}
to find the minimum cost assignment vector less than or equal to $T$ \label{fig:alg-price}}
\end{figure}

It follows that the set of constrained market clearing price vectors is closed under meets. By applying 
the lattice-linear predicate detection, we get an algorithm to compute the least constrained market clearing price shown in 
 Fig. \ref{fig:alg-price}.  In conjunction with Lemma \ref{lem:cons}, we get a
 generalization of Demange, Gale and Sotomayor's exact
auction mechanism \cite{demange1986multi} to incorporate lattice-linear constraints on the market clearing price.

In Fig. \ref{fig:alg-price}, we have used $\alpha(\C, j)$ as simply one unit of price.
For any item $j$ that is part of a minimal over-demanded set of items, we can increase its price by the minimum amount 
to ensure that some bidder $b$ can switch to her second most preferred item.

\section{Computing All Constrained Stable Matchings}
\label{sec:slice}

We now consider the problem of computing all constrained stable matchings. Since the number of stable matchings may be exponential in 
$n$, instead of keeping all matchings in explicit form, we would like a concise representation of polynomial size that
can be used to enumerate all constrained stable matchings. 
In SMP literature, rotation posets are used to capture all stable matchings.
We give a method based on a constructive version of Birkhoff's Theorem that can be
used to capture all stable matchings that satisfy external constraints  \cite{Birk3, mittal2001computation}.
A rotation poset \cite{gusfield1989stable} is a special case of our method when the set of external constraints is empty.
%
\remove{
The idea that Birkhoff's theorem is applicable for computing a concise representation of all matchings is explicitly mentioned 
in \cite{gusfield1989stable}. However, they state that ``{\em For algorithmic applications, Birkhoff's theorem is not the ideal tool to use to develop the structure of stable matching.
In short, while Birkhoff'theorem shows the existence of a partial order representing the distributive lattice, the theorem by itself does not suggest when or how the partial
order can be computed efficiently from the problem input.}'' 
We show that, on the contrary,  the lattice-linear predicate detection algorithm allows efficient construction of the
partial order to concisely represent the set of all (constrained) stable matchings.
}

We first define the dual concept of a lattice-linear predicate.
Just as a lattice-linear predicate allows us to start with the bottom element of the lattice and advance in the forward direction,
its dual allows us to start with the top element and advance in the backward direction.
Given any distributive lattice $\CC$ of $n$-dimensional vectors, and any predicate $B$, we say
  $ \mbox{reverse-forbidden}(\C,i, B) \equiv \forall H \in \CC : H \leq \C: (G[i] = H[i]) \Rightarrow
  \neg B(H).$
Observe that for {\em forbidden}, we considered $H \geq \C$, whereas for reverse-forbidden, we consider $H \leq \C$.
We define a predicate $B$ to
be {\em post-lattice-linear}  if for any $\C \in \CC$,
 $B$ is false in $\C$ implies that $\C$ contains a
{\em reverse-forbidden state}. Formally,
A boolean predicate $B$ is {\em {post-lattice-linear}} 
iff:
$\forall \C \in \CC: \neg B(\C) \Ra \exists i: \mbox{reverse-forbidden}(\C,i,B)$

It can be shown that $B_{marriage}$ is not only lattice-linear but also post-lattice-linear.

Since constrained stable matching is a post-lattice-linear predicate, from the dual of Lemma \ref{lem:lin}(b) it follows that the {\em feasible set}, the set of assignments
satisfying $B_{marriage}$, is also closed under joins. Therefore, the feasible set forms a sublattice of the lattice of all assignments.
Since a sublattice of a distributive lattice
is also distributive, the set of assignments that satisfy constrained stable marriage
forms a finite distributive lattice.
From Birkhoff's theorem \cite{davey}
we know that a finite distributive lattice
can be equivalently represented using the poset of its
join-irreducible elements. 
We now show that join-irreducible elements of the feasible sublattice can be constructed efficiently (i.e.
without constructing the sublattice which may be exponential in $n$ in its size).
%
The set of all elements of $\CC$ satisfying $B$ can be generated as the order ideals
of the following poset $ (\{ J(B, e) | e \in E \}, \subseteq) $
where $J(B, e)$ is the minimum order ideal of $(E, \leq)$ that satisfies $B$ and contains $e$.
It can be verified that  $J(B, e)$ is a join-irreducible element and that every join-irreducible element is of this form.
Interested readers can find details in \cite{mittal2001computation}.
The poset $ (\{ J(B, e) | e \in E \}, \subseteq) $ is called a {\em slice} in \cite{mittal2001computation}.

Hence, to compute the slice, it is sufficient to give a procedure to compute $J(B, e)$.
To determine $J(B, e)$ it is sufficient to 
use the algorithm for detecting a lattice-linear predicate by using the following predicate for every $e$:
$B_e(G) \equiv B(G) \wedge (e \in G)$.
Since $B$ is a lattice-linear predicate, and the predicate $e \in G$ is also lattice-linear, $B_e(G)$ is also lattice-linear.
Therefore, by using the algorithm for the constrained stable marriage introduced in this paper, we also
get an algorithm to compute the slice.

We illustrate this procedure by computing $J(B, (P_1, w_2))$.
We start with $[w_1, w_3, w_4, w_2]$ since $J(B, (P_1, w_1)) = [w_1, w_3, w_4, w_2]$.
On changing the assignment of $P_1$ to $w_2$, we get
$[w_2, w_3, w_4, w_2]$. Since $w_2$ prefers $P_1$ to $P_4$, we advance on $P_4$ to get
$[w_2, w_3, w_4, w_4]$. Since $w_4$ prefers $P_3$, we advance again on $P_4$ to get
$[w_2, w_3, w_4, w_3]$. Since $w_3$ prefers $P_2$, we advance on $P_4$ to get
$[w_2, w_3, w_4, w_1]$. This is a constrained stable matching.

By computing the set $ \{ J(B, e) | e \in E \}$, in this manner, we can 
compute the slice for any constrained SMP graph.
An analogous method can be used to compute a concise representation of all constrained market clearing prices when
the prices are integral.
\section{Conclusions and Future Work}
\label{sec:conc}

We have shown that many combinatorial optimization problems can be solved by applying
our lattice-linear predicate algorithm $LLP$. Specifically, we have shown that more general versions
of the stable matching problem, the single source shortest path problem and the market clearing prices problem can be
solved using $LLP$. 

There are many future directions for this work. Are there efficient lattice-linear predicate detection based parallel algorithms for 
the max-flow problem? 
%
Can the properties of the lattices be exploited to speed up the computation?
Are there efficient methods when the feasibility predicate is not lattice-linear?



\section{Acknowledgements.}
I thank Calvin Ly and Rohan Garg for many discussions on this topic.
I also thank Calvin Ly for implementing a version of LLP algorithm in Java.

\bibliography{refs,refs2,refs3}

\end{document}